\documentclass[sigconf]{acmart}

\usepackage{balance}    % Balance last page
\usepackage{paralist}   % Compressed lists
\usepackage{xspace}     % Fix space in macros
\usepackage{units}     % nicer slanted fractions

\usepackage{type1cm}    % Type1 computer modern font
\usepackage{pifont}     % Symbols and dingbat fonts
\usepackage{bold-extra} % bold + {small capital, italic}
\usepackage{siunitx}    % \num for decimal grouping
\usepackage{microtype}  % compress text

\usepackage{booktabs}   % Nicer tables
\usepackage{multirow}   % Multi rows for tables
\usepackage{graphicx}   % Advanced figures
\usepackage[font={bf}, tableposition=top]{caption}     % captions on top for tables
\usepackage[vlined,linesnumbered,ruled,noend]{algorithm2e}     % algorithms
\usepackage[font={small}]{subfig} % subfig, 4 figures in a row
\usepackage{footnote} % show footnotes in tables
\makesavenoteenv{table}

\usepackage{mathtools}  % amsmath++

\usepackage[hide]{chato-notes}
\newcommand{\hide}[1]{}

\newcommand{\bonus}{{\sc Bonus}\xspace}
\newcommand{\bonusValue}{{\ensuremath{\mathrm{b}}}\xspace}
\newcommand{\bonusInc}{{\ensuremath{\mathrm{{\delta}{b}}}}\xspace}
\newcommand{\optBonusValue}{\ensuremath{\mathrm{b_{_{OPT}}}}\xspace}
\newcommand{\bonusPolicy}{\ensuremath{\mathbf{B}}\xspace}
\newcommand{\quota}{{\sc Quota}\xspace}
\newcommand{\quotaValue}{{\ensuremath{\mathrm{q}}}\xspace}
\newcommand{\merit}{{\sc Coefficients}\xspace}
\newcommand{\meritPolicy}{\ensuremath{\mathbf{M}}\xspace}
\newcommand{\fair}{{\sc FA*IR}\xspace}
\newcommand{\median}{{\sc MEDIAN}\xspace}

\newtheorem{problem}{Problem}
\newtheorem{lemma}{Lemma}

\newcommand{\prob}[1]{\ensuremath{\mathbf{P}(#1)}\xspace}

\newcommand{\expect}[1]{\ensuremath{\mathbb{E}[#1]}}
\newcommand{\scoringFunction}[1]{\ensuremath{\mathbf{f}(#1)}\xspace}
\newcommand{\scoringFunctionInstance}{\ensuremath{\mathbf{f}}\xspace}

\newcommand{\cumul}[1]{\ensuremath{F(#1)}\xspace}
\newcommand{\cumulValue}{\ensuremath{\beta}\xspace}
\newcommand{\invcumul}[1]{\ensuremath{F^{-1}(#1)}\xspace}

\newcommand{\rotate}[1]{\ensuremath{{\rho}(#1)}\xspace}

\newcommand{\optFunction}[1]{\ensuremath{o(#1)}\xspace}

\newcommand{\binarysearch}[1]{\ensuremath{\mathrm{binary\-search}(#1)}\xspace}

\newcommand{\direction}{\ensuremath{\mathbf{d}}\xspace}

\newcommand{\interval}{\ensuremath{\mathbf{I}}\xspace}
\newcommand{\sensitive}{\ensuremath{\mathbf{A}}\xspace}
\newcommand{\sensitiveGroup}{\ensuremath{\mathbf{a}}\xspace}
\newcommand{\schoolgrades}{\ensuremath{\mathbf{X}}\xspace}
\newcommand{\schoolgradesValue}{\ensuremath{\mathbf{x}}\xspace}

\newcommand{\score}{\ensuremath{\mathbf{S}}\xspace}
\newcommand{\scoreValue}{\ensuremath{\mathbf{s}}\xspace}

\newcommand{\scoreThreshold}{\ensuremath{\mathrm{\tau}}\xspace}
\newcommand{\admission}{\ensuremath{\mathbf{T}}\xspace}
\newcommand{\admissionFunction}[1]{\ensuremath{\mathbf{t}(#1)}\xspace}
\newcommand{\admissionFunctionInstance}{\ensuremath{\mathbf{t}}\xspace}
\newcommand{\admissionFraction}{\ensuremath{\mathrm{\theta}}\xspace}
\newcommand{\excellence}{\ensuremath{\mathbf{Y}}\xspace}

\newcommand{\unavailable}{\ensuremath{\mathbf{n/a}}\xspace}
\newcommand{\uoa}{\ensuremath{\mathbf{UoS}}\xspace}

\newcommand{\dmd}{\ensuremath{\mathbf{DmD}}\xspace}

\newcommand{\schoolgradesWeight}{\ensuremath{\mathbf{w}}\xspace}
\newcommand{\someIntercept}{\ensuremath{{\alpha_0}}\xspace}
\newcommand{\normalization}{\ensuremath{{\alpha}}\xspace}
\newcommand{\bestWeights}{\ensuremath{\mathbf{c}}\xspace}
\newcommand{\optWeights}{\ensuremath{\mathbf{c_{_{\text{opt}}}}}\xspace}

\newcommand{\spara}[1]{\smallskip\noindent{\bf{#1}}}
\newcommand{\mpara}[1]{\medskip\noindent{\bf{#1}}}

\newcommand{\squishlist}{\begin{list}{$\bullet$}
  { \setlength{\itemsep}{0pt}
     \setlength{\parsep}{3pt}
     \setlength{\topsep}{3pt}
     \setlength{\partopsep}{0pt}
     \setlength{\leftmargin}{1.5em}
     \setlength{\labelwidth}{1em}
     \setlength{\labelsep}{0.5em} } }
\newcommand{\squishend}{
\end{list}  }

\copyrightyear{2020}
\acmYear{2020}
\setcopyright{acmlicensed}
\acmConference[SAC '20]{The 35th ACM/SIGAPP Symposium on Applied Computing}{March 30-April 3, 2020}{Brno, Czech Republic}
\acmBooktitle{The 35th ACM/SIGAPP Symposium on Applied Computing (SAC '20), March 30-April 3, 2020, Brno, Czech Republic}
\acmPrice{15.00}
\acmDOI{10.1145/3341105.3373878}
\acmISBN{978-1-4503-6866-7/20/03}

\newcommand{\ourtitle}{Affirmative Action Policies for Top-k Candidates Selection}
\newcommand{\oursubtitle}{With an Application to the Design of Policies for University Admissions}
\title{\ourtitle}
\subtitle{\oursubtitle}

\author{Michael Mathioudakis}
\affiliation{%
  \institution{University of Helsinki}
}
\author{Carlos Castillo}
\affiliation{%
  \institution{Universitat Pompeu Fabra}
}
\author{Giorgio Barnabo}
\affiliation{%
  \institution{Sapienza University of Rome}
}
\author{Sergio Celis}
\affiliation{%
  \institution{Universidad de Chile}
}
\renewcommand{\shortauthors}{Mathioudakis, Castillo, Barnabo, Celis}

\begin{abstract}
We consider the problem of designing affirmative action policies for selecting the top-k candidates from a pool of applicants.
We assume that for each candidate we have socio-demographic attributes and a series of variables that serve as indicators of future performance (e.g., results on standardized tests) -- as well as historical data including the actual performance of previously selected candidates. % who were selected under the current policy.
We consider the case where
an organization wishes to increase the selection of people from disadvantaged socio-demographic groups.
Hence, we seek to design an \emph{affirmative action} policy to select candidates who are more likely to perform well, but in a way that increases the representation of disadvantaged groups.

Our motivating application is the design of university admission policies to bachelor's degrees.
We use a causal framework to describe several families of policies (changing component weights, giving bonuses, and enacting quotas), and compare them both theoretically and through extensive experimentation on a real-world dataset containing thousands of university applicants.
Our empirical results indicate that simple policies could favor the admission of disadvantaged groups without significantly compromising on the quality of accepted candidates.
\end{abstract}

\ccsdesc[300]{Computing methodologies~Model development and analysis}
\ccsdesc[300]{Applied computing~Computers in other domains}

\keywords{algorithmic fairness, data-driven policy design}

\begin{document}

\maketitle

%!TeX root = admissions.tex

\section{Introduction}
\label{sec:introduction}

The emergent academic field of algorithmic fairness is concerned with the impact of algorithmic decision making on society, particularly over disadvantaged groups.
%
% In recent years, researchers across several areas of computing, including data mining, machine learning, and information retrieval, have sought to develop methods that are non-discriminatory, i.e., that do not generate an unjustified disadvantage for members of a socially salient group.

\spara{Predictive policies for top-k selection.}
We consider a general setting in which we have a large pool of applicants for, e.g., an educational program, a scholarship, or a job, and we want to select the $k$ most promising ones.
This selection is based on the candidates' predicted future performance.
For instance, in the case of university admissions, this prediction stems from previous grades obtained in different subjects (transcripts), scores in relevant standardized tests, ratings from interviews with the candidate, and so on.
The predictor, which induces a \emph{selection policy}, is created from historical data about the academic performance of previously selected candidates.
A key constraint is that such outcomes are only observed for candidates that have been selected in the past -- but not for rejected ones.

\spara{Affirmative action in top-k selection.}
Algorithmic fairness is commonly operationalized through a notion of \emph{demographic parity} or \emph{statistical parity}, which requires that a decision such as accepting or rejecting a candidate is independent of the sensitive attribute
(for a criticism of this notion see \cite{dwork2012fairness}).
In our setting, for example, we want to admit candidates of different groups at similar rates.
%
%In some cases we might not seek equality, but a reduction of inequality; in any case this notion is more commonly known as \emph{affirmative action} or \emph{positive action}.
%
Given that disparities exist, we seek an \emph{affirmative action} policy: a temporary intervention aimed at increasing the representation of an under-represented group.

\spara{University admissions.}
The problem domain we use throughout this paper is university admission policies.
In several countries, standardized tests for university admission are administered at the end of the last year of high school or community college.
They are often framed as a measurement of aptitude for university studies, i.e., a prediction of how a student would eventually perform if admitted.
For most students, admission decisions are based on an \emph{admission score} computed from test scores and/or high school grades (e.g., through a linear combination).
 % or other measures of previous academic achievement.

Admission scores offer one mechanism for admission targeted at aptitude.
However, universities may also want to have a diverse body of students, or may wish to promote gender equality and social mobility.
For this reason, university admissions often include a predictor of students' future performance plus a patchwork of additional affirmative action programs addressed at increasing the number of accepted candidates from disadvantaged backgrounds. %, as we explain next.

% It is thus reasonable to ask whether one can design admission policies with certain desirable properties.
% 150 million students in a Bachelor degree in 2017
% http://data.uis.unesco.org/
% ISCED 6 programs both sexes"

\spara{Designing affirmative action policies.}
The goals of an affirmative action top-k selection policy are twofold.
First, it should select candidates who have a high expected performance.
Second, it should ensure a sufficient number of candidates from disadvantaged backgrounds are selected.
Both properties can be seen as expressing notions of fairness: the former towards good performance; the latter against disadvantages that are not related with performance.

Effectively, to design such a policy one needs to predict its effect before it is enacted.
%
%Suppose a policy adds a bonus to the score of candidates from a disadvantaged background.
%
%How many more candidates from that background we expect to selected? And how will they perform if selected?
%
In this paper, we formalize the problem of designing data-driven affirmative action policies and to develop algorithms to build such policies starting from historical data.
Our contributions have direct application to university admissions, but they are directly extensible to other domains.

\spara{Our approach.}
To formalize the questions we consider, we build a probabilistic causal model that captures the selection process. % and makes our assumptions about the underlying data explicit.
The causal model takes into account each applicant's input attributes (e.g., test scores), as well as a number of {\it sensitive} socio-demographic attributes, to compute an admission score, which is the main basis for their selection.
Given the causal model, we consider three types of policies:
\begin{inparaenum}
\item \merit: policies that adjust the weight of different components of the admission score to emphasize those that favor disadvantaged groups;
\item \bonus: policies that give extra points in the admission score to applicants from disadvantaged groups; and
\item \quota: policies that enforce a certain selection rate for applicants from disadvantaged groups.
\end{inparaenum}
Policies are considered from the perspective of a single institution and optimized to increase the expected performance of the selected applicants while accepting a suitable proportion of disadvantaged applicants.

In our analysis, we build instances of causal models from the available data, find suitable parameters for a series of affirmative action policies,
%
%Furthermore, we discuss how our algorithms extend to simpler or more elaborate variants of the causal model.
%
and present the results of an empirical study over university admissions considering thousands of applicants.

%!TeX root = admissions.tex
\section{Related Work}\label{sec:related}

\subsection{Fairness-aware top-k selection}

% This paper belongs to the emergent research field of \emph{algorithmic fairness}, specifically to fairness in predictive decision making (for a survey, see~\cite{mitchell2018prediction}).
%
% Historically, data mining and machine learning applications have been the focus of most works, with comparatively less research dedicated to fairness in ranking or top-k selection.
%
Ranking with fairness constraints, in batch processing and online processing scenarios respectively, is the focus of two recent works by \citet{celis2017ranking} and \citet{stoyanovich2018online}.
In these works, every element to be ranked has a certain utility and belongs to a certain group, and for every ranking position $r=1,2,\dots,k$ there are two constraints per group indicating the minimum and maximum number of candidates from that group that must be among the top-$r$ positions.
The goal is to rank candidates by decreasing utility without violating the constraints.

\citet{zehlike2017fa} define a \emph{fair representation condition} based on a binomial test that seeks to determine if a ranking of $k$ elements is compatible with a process in which candidates have been drawn at random from protected and non-protected groups with a known probability.
The fair representation condition together with a process for adjusting for multiple hypothesis testing can be used to obtain the bounds used in ranking with fairness constraints.
\citet{kearns2017meritocratic} study the case of selecting $k$ candidates drawn from $d$ different groups in which candidates across groups are not comparable.
Hence, what matters is the relative position of candidates within their groups.

\hide{
When a ranking function is employed, one approach to ensure that equal representation of groups is to use ``repaired'' values for the features of the ranked items, such that the repaired feature distributions of the two groups are the same.
In this spirit, \citet{feldman2015certifying} offer algorithms to obtain such repaired values, in a way that the ranking within each group is preserved, while it is impossible to distinguish the group of a ranked item from its repaired features.
}

\subsection{Affirmative action}

\spara{``Downstream'' effects of affirmative action.}
Several recent papers study ``downstream'' or ``cascade'' effects of affirmative action policies, mostly through a theoretical approach.
\citet{hu2018short} consider a labor market in which firms are either temporary or permanent employers. Temporary employers are forced to use hiring policies that favor disadvantaged groups, while permanent employers can hire based on expected utility alone. The goal is to increase the employability of candidates from disadvantaged groups through their employment in the temporary market.
\citet{kannan2019downstream} analyze the potential effects of affirmative action in education on the employability of graduates from disadvantaged groups. The goal is to achieve parity of employability or create an incentive for employers to adopt a color-blind hiring policy.
\citet{mouzannar2019fair} analyze the effect of affirmative action policies on the qualifications of different groups in society.

\spara{Data-driven analysis.}
Over the last four decades prestigious universities, many of which are public institutions, have become even more selective, having a negative effect on students from low socio-economic backgrounds and underrepresented minorities \cite{posseltetal2012access}.
This effect is mainly explained by the lower performance of disadvantaged students on standardized tests and other admission requirements that favor applicants from well-off families \cite{bastedobowman2017admission}.
However, admission scores might underestimate university performance for disadvantaged students.
For instance, a study by \citet{wightman1998lsac} on nearly 25,000 law students in the US found that while students of color had significantly poorer results on admission tests than white students, their difference in terms of performance once they were admitted (their probability of passing the bar exam) did not justify the large gap between the admission rates of both groups.
%
%``Although students of color entered law school with academic credentials, as measured by UGPA and LSAT scores, that were significantly lower than those of white students, their eventual bar passage rates justified admission practices that look beyond those measures.''

As a result, countries and universities have promoted affirmative action policies. % to increase diversity and social representation.
For instance, Brazil has promoted bonus \cite{estevan2018redistribution} and racial quotas \cite{francistannuri2012brazilaffirmativeaction}.
The evidence suggests that the implementation of these policies increased both black students and students coming from low socioeconomic backgrounds in large public universities \cite{francistannuri2012brazilaffirmativeaction}.
In Chile, where significant gaps are observed in standardized math tests between students of public versus private schools \cite{cornejo2006experimento} and between women and men \cite{ariasetal2016genderadmision}, coefficient-based policies have also shown some effectiveness \cite{larroucauetal2015rankinadmission}.
\hide{
Institutional policies seem to have an effect on admission systems that include socio-demographic information in the selection process, as done by several universities in the US.
For instance, \citet{bastedoetal2018admission} found that when admission officers in the US consider contextual variables (economic hardship, family background) into the admission decision, the chances of admitting disadvantaged students increases.
}

\subsection{Causal reasoning and fairness}

\citet{hardt2016equality} explain how causal reasoning, as developed by Pearl~\cite{bookofwhy} and others, offers a theoretical framework to express notions of algorithmic fairness that go beyond statistical / observational measures.
Indeed, the work of \citet{kusner2017counterfactual} and \citet{kilbertus2017avoiding} builds upon that intuition by studying notions of \emph{counterfactual fairness} that take into account the causal structure of the data's generative model.
In our work, we use a causal framework to define the model for the admission mechanism (Section~\ref{sec:setting}).
% , and confirm that the available data allow us to make unbiased evaluation of admission policies (Section~\ref{sec:dataanalysis}).

%!TeX root = admissions.tex
\section{Setting}
\label{sec:setting}

In the setting we consider, selection of candidates is performed by \emph{one institution}.
We begin by defining a probabilistic causal model that captures this process, expresses our assumptions about the data, and allows us to quantify the effect of different selection policies (Section~\ref{sec:causalmodel}).
For a given policy, we then define measures of \emph{utility,} capturing the performance of selected candidates, and \emph{fairness,} capturing the similarity between selection rates (Section~\ref{sec:measures}).
Next, we describe three types of policies (Section~\ref{sec:admissionpolicies}), and present a formal definition of the technical problem we address as a trade-off between utility and fairness (Section~\ref{sec:problemdefinition}).
Finally, we discuss various modeling choices we make (Section~\ref{sec:modeldiscussion}).

\subsection{Causal model}
\label{sec:causalmodel}

Our design of selection policies is guided by historical data that are the result of a particular enacted selection policy.
To predict what data we would observe under alternative policies, it is necessary to express formally our knowledge and assumptions about the decision mechanisms we study.
We do this by defining a \emph{causal diagram}~\cite{bookofwhy} capturing our background knowledge about causal relationships between quantities of interest, as shown in Figure~\ref{fig:basicdiagram}.
The diagram describes the selection mechanism from the point of view of \emph{one candidate} to \emph{one institution}, and assumes that each candidate is associated with socio-economic attributes \sensitive, which are considered sensitive, and a series of input scores \schoolgrades (e.g., results of qualification tests), which are considered non-sensitive.
These are combined into a selection score \score which is the basis for the institution's decision \admission.
In particular, given input scores \schoolgrades = \schoolgradesValue and socio-economic attributes \sensitive = \sensitiveGroup, the candidate obtains a selection score \score = \scoreValue according to %a mathematical formula \scoreValue = \scoringFunction{\schoolgradesValue, \sensitiveGroup}.
$
	\score = \scoringFunction{\schoolgrades, \sensitive}.
% \label{eq:affirmativescore}
$
Here, we allow the selection score \score to depend on \sensitive to allow for affirmative action through the selection score function.
Based on this score \score a decision \admission is made in favor or against the candidate.
Again, we allow \admission to be determined by \sensitive so as to allow also for affirmative action through policies that operate at the stage of decision.

\begin{figure}[t]
\begin{center}
\subfloat{
\includegraphics[width=0.35\columnwidth]{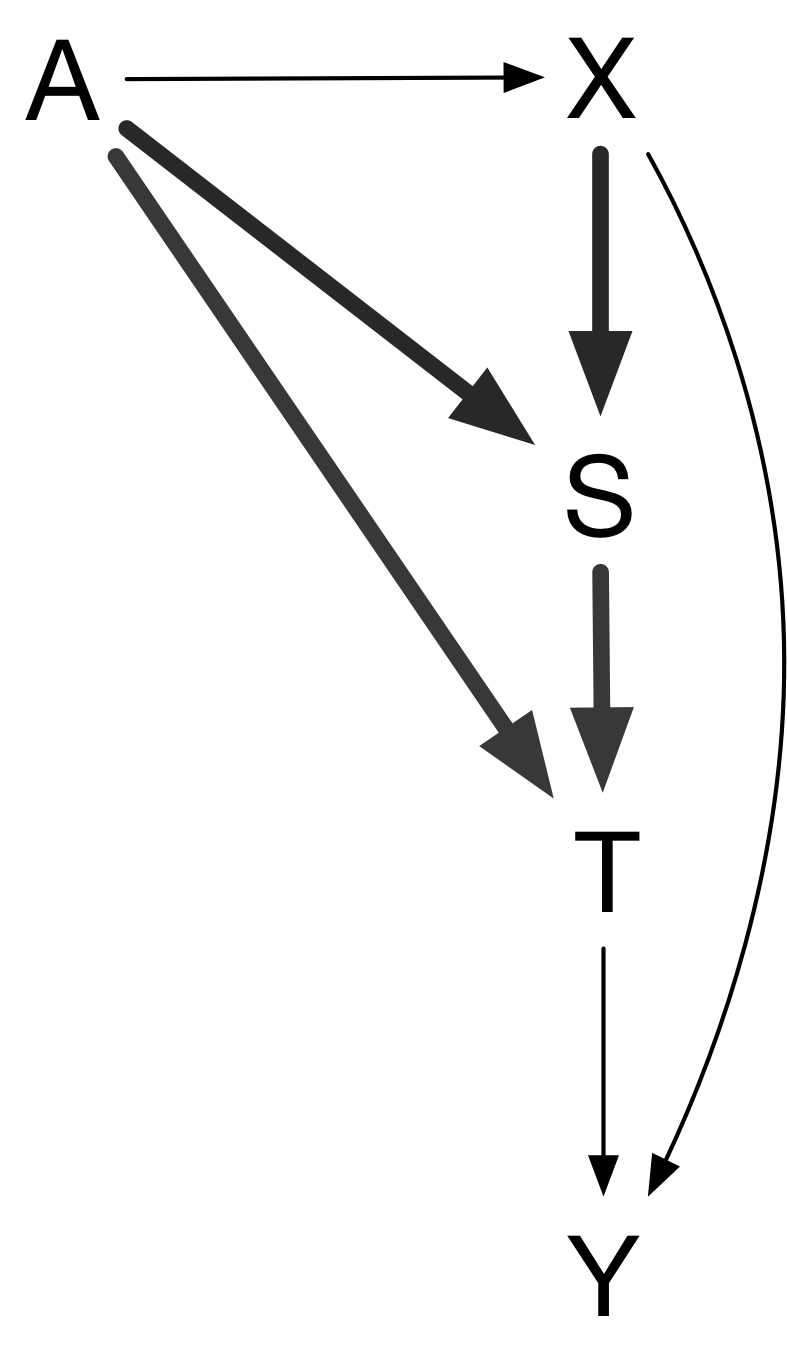}
}

\subfloat{
	\begin{tabular}[t!]{cl}\toprule
	Symbol & Description \\ \midrule
	\sensitive & Socio-economic attributes \\
	% \preferences & applicant preferences \\
	\schoolgrades & Input scores (e.g., results of qualification tests)\\
	\score & Selection score: combination of input scores\\
	% \rank & applicant ranking \\
	\admission & Decision of the institution: select or not select \\
	\excellence & Performance of selected candidate \\ \bottomrule
	\end{tabular}
}
\end{center}
\caption{Diagram of causal model. Selection policies are defined via the causal relationships $(\sensitive, \schoolgrades)\rightarrow\score$ and $(\sensitive, \score)\rightarrow\admission$.
Relationships affected by the affirmative action policy are shown with thick arrows.}
\label{fig:basicdiagram}
\end{figure}

For the purposes of affirmative action policy design, it is also crucial that we model the possibility that sensitive socio-economic attributes \sensitive influence the input scores \schoolgrades.
This is done through the dependency $\sensitive\rightarrow\schoolgrades$ in the diagram.
%
% In particular, it is possible that low socio-economic status may cause low input scores.
%
This is the case if, for instance, university applicants from families of higher income have access to training or resources that are not affordable for families of lower income.
% and low performance of admitted students;
In that case, an affirmative action policy tries to bring \emph{equality of opportunity} \cite{roemer1998equality}.

The institution under consideration is interested in the acceptance of only a certain number or fraction of top candidates.
Formally, we require every selection policy to be \emph{calibrated} -- i.e., that it lead, in expectation, to the selection of a certain fraction \admissionFraction of candidates, considered known and fixed for the institution.
\begin{equation}
\expect{\admission} = \prob{\admission = 1} = \admissionFraction \label{eq:calibration}
\end{equation}

If the decision for a candidate is positive ($\admission = 1$), we observe their performance \excellence.
If a candidate is rejected the performance of that candidate is not observed, and we write $\excellence = \unavailable$.

Finally, to simplify the presentation of our analysis, we will always assume implicitly that continuous variables in our model are associated with well-behaved density functions -- in particular, with continuous cumulative distributions.

\subsection{Measures}
\label{sec:measures}

We aim to define selection policies that lead to the selection of the candidates with the highest expected performance.
At the same time, we aim to mitigate the disadvantage of certain demographic groups.
Towards this end, we define the following measures.

\spara{Utility of Selection (\uoa).}
A policy is useful if it leads to the selection of the candidates that actually perform best.
This is captured by the following measure, defined in terms of the causal model.
\[
\uoa = \expect{\excellence | \admission = 1}
\]

\spara{Demographic Disparity (\dmd).}
Our objective is to produce selection policies that minimize the difference between the selection rates of candidates of different backgrounds, according to their socio-economic features \sensitive.
The following measure is defined for a group of candidates with sensitive properties \sensitive = \sensitiveGroup; it compares the selection rate for this group to the candidates outside the group.
\[
\dmd_{_\sensitive} = {\prob{\admission = 1 | {\sensitive = \sensitiveGroup}}} - {\prob{\admission = 1 | {\sensitive \not= \sensitiveGroup}}}
\]

The demographic disparity is equivalent to the \emph{risk difference} measure from the literature \cite{donini2018empirical}, where risk is the probability of not being selected $1-\prob{\admission=1}$, and the difference is taken between the sensitive group and its complement.

\subsection{Types of selection policies}
\label{sec:admissionpolicies}

To define alternative selection policies, we essentially have to define instances of the relationships $(\sensitive, \schoolgrades)\rightarrow\score$ and $(\sensitive, \score)\rightarrow\admission$, marked with thick arrows in Figure~\ref{fig:basicdiagram}.
In other words, for a given causal model, we aim to define the scoring function \score = \scoringFunction{\sensitive, \schoolgrades} and selection function \admission = \admissionFunction{\sensitive, \score} so that the policy exhibits desired values of the utility and fairness measures as described earlier.
We say that every pair of functions \scoringFunction{\sensitive, \schoolgrades} and \admissionFunction{\sensitive, \score} defines a \emph{selection policy} within the \emph{selection mechanism} defined by the model.
We consider three types of policies: \merit, \bonus, and \quota.

\spara{\merit-based policies}. This type of policy is completely determined by a series of coefficients \schoolgradesWeight that are used to compute the selection score \score as a linear combination of the input scores \schoolgrades.
No sensitive attributes are used.
The score is computed as
\begin{equation}
	\score = \scoringFunction{\schoolgrades = \schoolgradesValue} = \schoolgradesWeight\cdot\schoolgradesValue,
		\;\;\; |\schoolgradesWeight|_{_1} = 1.
		\label{eq:meritscoring}
\end{equation}
Note that \schoolgradesWeight is normalized. Each component of \schoolgradesWeight represents the degree to which an input score counts towards the selection score.
Candidates are admitted if their score exceeds a threshold \scoreThreshold, in which case we write $\admission = 1$ --- otherwise, we have $\admission = 0$.
\begin{equation}
\admission = \admissionFunction{\score} =
	\begin{cases}
		1\ \text{if}\ \score \geq \scoreThreshold \\
		0\ \text{otherwise}
	\end{cases} \label{eq:meritadmission}
\end{equation}

\spara{\bonus policies}.
This type of policy differs from \merit policies in that sensitive attributes \sensitive contribute to the selection score of candidates.
Specifically, the group of candidates with \sensitive = \sensitiveGroup receive an additive bonus \bonusValue in their score.
\begin{equation}
	\score = \scoringFunction{\schoolgrades = \schoolgradesValue, \sensitive} =
		\begin{cases}
			\schoolgradesWeight\cdot\schoolgradesValue + \bonusValue,\hfill  \text{if}\ \sensitive = \sensitiveGroup \\
			\schoolgradesWeight\cdot\schoolgradesValue,\quad\quad  \text{if}\ \sensitive \not= \sensitiveGroup
		\end{cases}, \quad |\schoolgradesWeight|_{_1} = 1.
\end{equation}
Note that the selection decision \admission is determined in exactly the same way as for \merit policies, i.e., via a threshold \scoreThreshold.

\spara{\quota policies}.
This type of policy determines the selection score in the same way as \merit, i.e., via a linear function with normalized weights \schoolgradesWeight.
However, unlike \merit policies, sensitive attributes \sensitive contribute to the decision \admission.
\quota policies constraint the selected candidates from group $\sensitive = \sensitiveGroup$ to be (in expectation) a fraction $\quotaValue\in[0,1]$ of accepted candidates.
% \note[ChaTo]{I think it is more natural to define the quota in terms of the group \sensitiveGroup (instead of its complement), so that they are a fraction \quotaValue.} % -- Done-MM
%
\begin{equation}
\prob{\admission = 1, \sensitive = \sensitiveGroup} = \quotaValue\cdot\admissionFraction \label{eq:quota}
\end{equation}
% \note[Giorgio]{what is theta? shouldn't it be given T=1?} -- It's defined above and we have defined it in this way consistently in the paper, so let's keep it as is
%
Specifically, the selection decision \admission under a \quota policy is determined based on separate thresholds
$\scoreThreshold_{_{(\sensitive \not= \sensitiveGroup)}}$ and
$\scoreThreshold_{_{(\sensitive = \sensitiveGroup)}}$
for the two groups, $(\sensitive=\sensitiveGroup)$ and $(\sensitive\not=\sensitiveGroup)$, respectively.
\begin{equation}
\admission = \admissionFunction{\score, \sensitive} =
	\begin{cases}
		1\ \text{if}\ (\sensitive\not=\sensitiveGroup)\ \text{and}\ \score \geq \scoreThreshold_{_{(\sensitive \not= \sensitiveGroup)}} \\
		1\ \text{if}\ (\sensitive=\sensitiveGroup)\ \text{and}\ \score \geq \scoreThreshold_{_{(\sensitive = \sensitiveGroup)}} \\
		0\ \text{otherwise}
	\end{cases} \label{eq:quotaadmission}
\end{equation}
The two thresholds are chosen so that Equation~\ref{eq:quota} is satisfied.
%
%The case where we want equal selection rates for people in different groups is a special case, and given we know the number of candidates from each group is trivial to implement.

\subsection{Problem definition}
\label{sec:problemdefinition}

Having defined the set of policies that we consider, we can now define our objective formally.
If we were to ignore the goal of favoring disadvantaged group $(\sensitive=\sensitiveGroup)$ via an affirmative action policy, it would be natural to aim for the policy that optimizes the utility \uoa of the selection mechanism.
On the other extreme, if we were to target only the demographic disparity of those selected, it would be natural to aim for a policy that eliminates \dmd (i.e., leads to $\dmd = 0$).
In general, we consider cases where we are willing to trade $\lambda$ units of utility to decrease disparity by one unit.
We thus define our technical objective as follows.
\begin{problem}
For sensitive property \sensitive, and two associated groups of candidates $(\sensitive=\sensitiveGroup)$ and $(\sensitive\not=\sensitiveGroup)$,
define \scoringFunction{\sensitive, \schoolgrades} and \admissionFunction{\sensitive, \score} so as to obtain a calibrated policy that maximizes
\[
\optFunction{\scoringFunctionInstance,\admissionFunctionInstance} = \uoa(\scoringFunctionInstance,\admissionFunctionInstance) - \lambda\ |\dmd_{_\sensitive}(\scoringFunctionInstance,\admissionFunctionInstance)|,\;\;\lambda \geq 0.
\]
\label{problem:utilitydisparity}
\end{problem}
Notice that the objective function considers the absolute value of \dmd, as we wish to eliminate \dmd rather than maximize it in favor of the currently disadvantaged group.

\subsection{Discussion on modeling choices}
\label{sec:modeldiscussion}
The causal model described above might not match exactly every real selection mechanism. %, as it is based on certain  assumptions, some made for the purposes of simplification.
We briefly discuss its assumptions in the context of university admissions. %, in which we apply our data analysis (Section~\ref{sec:dataanalysis}).

\spara{Candidate decisions might not be independent.}
In our setting, the decision is made independently for each applicant based on their individual attributes.
This modeling choice simplifies the analysis considerably, while it still allows us to understand the properties of the different policies.
However, it does not hold in some cases.
For example, for many universities there is a limit on the number of admitted students -- therefore admissions are \emph{not independent}, as one candidate's admission can be another's rejection.

\spara{Selection does not imply enrollment.}
We also assume that an applicant who is admitted by a university is bound to attend it.
However, in reality applicants do not apply for admission to one institution only and might be admitted by several.
In some systems, applicants are forced to prioritize their preferences for institutions and attend the first to admit them.
In other systems, applicants can choose which one to attend.
The model does not capture explicitly such admission systems.
However, its assumptions work well for top universities that students are almost certain to attend if admitted.
This is the case for our empirical analysis, which concerns a top university to which most selected applicants enroll (about $95\%$).
\hide{
We also acknowledge that highly selective universities with affirmative action policies may create a mismatch issue (i.e., accepted students may be better off in another type of institution) as reported in countries such as the US and India \cite{frisanchokrishna2015missmatchindia}.
}
%!TeX root = admissions.tex
\section{Policy search for a single sensitive attribute}

This section describes algorithms to determine affirmative action policies for a single sensitive attribute \sensitive.
Section~\ref{sec:multiplesensitiveattributes} extends them to multiple sensitive attributes.

\subsection{\merit policy}
\label{sec:algorithmsformerit}

Let us consider the case of $\lambda = 0$, in which the objective function simplifies to the measure of utility \uoa.
Let  $m(\schoolgradesValue) = \expect{\excellence | \schoolgrades = \schoolgradesValue, \admission = 1}$ be the expected performance \excellence of a selected applicant with input scores \schoolgradesValue.
By straightforward calculations, we find that
\begin{equation}
	\uoa = \expect{\excellence | \admission = 1} =
		\frac{1}{\admissionFraction}
		\int m(\schoolgradesValue)\ \delta(\scoringFunction{\schoolgradesValue}\geq\scoreThreshold)\
		d\prob{\schoolgradesValue}.
		\label{eq:expandedUoA}
\end{equation}
Intuitively, Eq.~\ref{eq:expandedUoA} expresses that \uoa is the average expected performance of applicants that exceed the threshold \scoreThreshold.
%
%It is easy to see from the equation that
Using $\scoringFunction{\schoolgradesValue} = m(\schoolgradesValue)$ maximizes \uoa over all possible scoring functions.
In other words, if we could use the expected performance of an applicant as the scoring function, the resulting \merit policy would select those applicants that are expected to have the best performance \excellence and it would be optimal for $\lambda = 0$.
However, we are constrained to use linear functions as scoring function \scoringFunction{\schoolgrades}, as defined in Equation~\ref{eq:meritscoring}.

In the case where $m(\schoolgradesValue)$ is indeed a linear function, then using $\scoringFunction{\schoolgradesValue} = m(\schoolgradesValue)$ leads to an optimal \merit policy  for $\lambda = 0$.
To state the claim formally, let \cumul{\score = \scoreValue} be the cumulative probability distribution for selection \score and \invcumul{\cumulValue} the inverse cumulative distribution, i.e., the function that returns the score below which lies a fraction \cumulValue of the population of applicants.

\begin{lemma}
Let
$
	m(\schoolgradesValue) = \someIntercept + \normalization\ \bestWeights\cdot\schoolgradesValue
$
with normalized coefficients, $|\bestWeights|=1$.
Then, the following \merit policy is optimal for $\lambda = 0$:
\begin{align}
	\score & = \scoringFunction{\schoolgrades, \sensitive} = \bestWeights\cdot\schoolgrades \nonumber\\
	\admission & = \admissionFunction{\score, \sensitive} =
	\begin{cases}
		1\ \text{if}\ \score \geq \scoreThreshold \\
		0\ \text{otherwise}
	\end{cases} \label{eq:goodmerit} \\
	\scoreThreshold & = \invcumul{1 - \admissionFraction}. \nonumber
\end{align}
\end{lemma}

Conceptually, the optimal \merit policy for $\lambda = 0$ corresponds to ordering applicants in decreasing order of $\bestWeights\cdot\schoolgrades$ and keeping the top \admissionFraction-fraction of them by setting the threshold \scoreThreshold appropriately.
By construction, this policy is calibrated.

For general $\lambda\geq 0$, \bestWeights is not necessarily optimal anymore -- instead, the optimal \merit policy would require different score weights.
To see why, notice that, for two data points $\schoolgradesValue_1$, $\schoolgradesValue_2$, we have
\begin{equation*}
	\schoolgradesWeight\cdot\schoolgradesValue_1 > \schoolgradesWeight\cdot\schoolgradesValue_2
	\;\;\text{iff}\;\; \|\schoolgradesValue_1\|\cos(\schoolgradesWeight, \schoolgradesValue_1) > \|\schoolgradesValue_2\|\cos(\schoolgradesWeight, \schoolgradesValue_2).
\end{equation*}
Therefore, the selection score depends on the angle between input scores \schoolgradesValue and score weights \schoolgradesWeight.
In general, a rotation of score weights \bestWeights by some angle \direction might lead to score weights $\schoolgradesWeight$ that change the ordering of datapoints $\{\schoolgradesValue\}$ by score value $\{\scoringFunction{\schoolgradesValue}\}$ -- and thus possibly lead to smaller \dmd and better objective value (see Figure~\ref{fig:angles}).

\begin{figure}
\includegraphics[width=0.60\columnwidth]{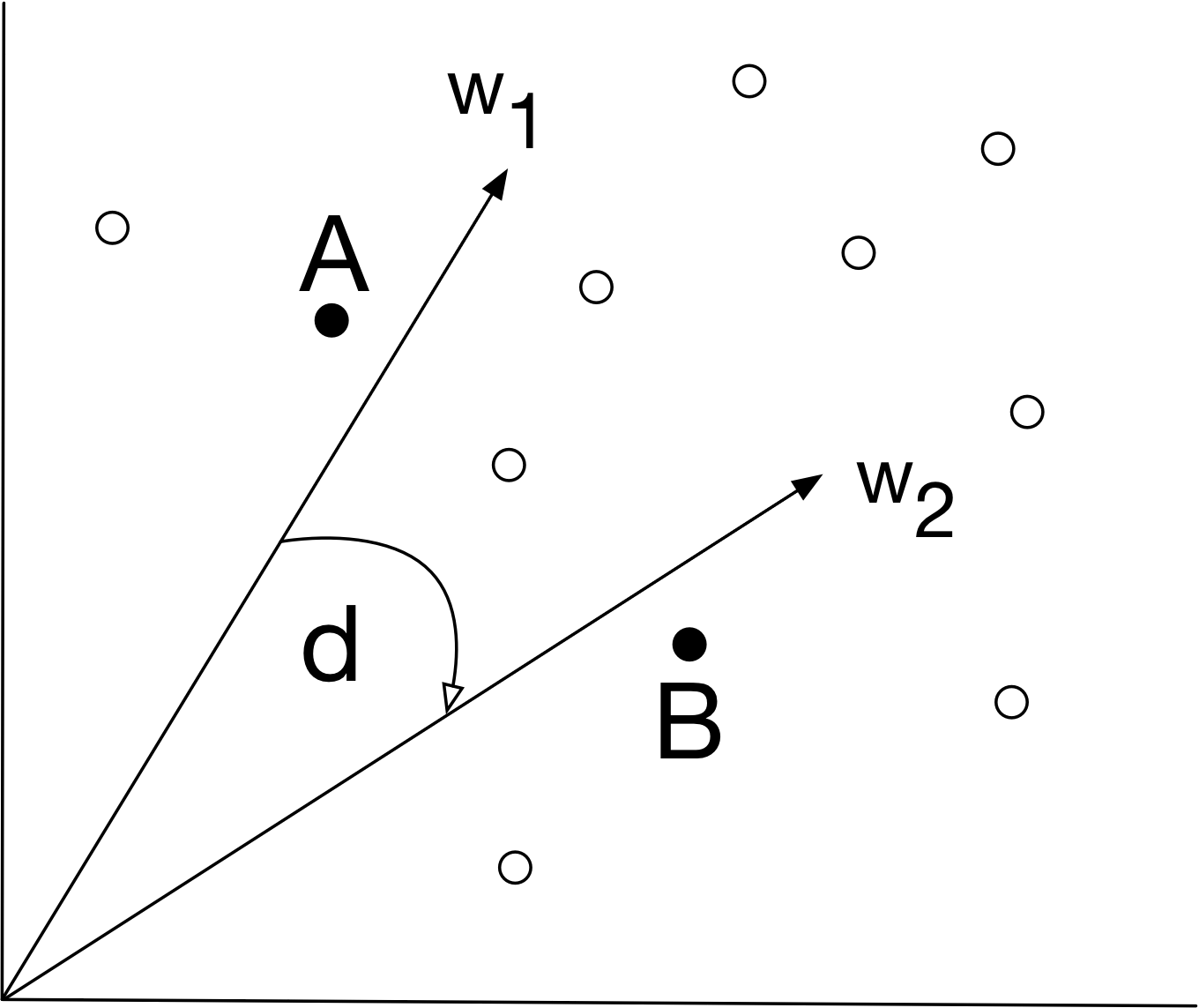}
\caption{Given a weight vector \schoolgradesWeight, the selection score of a candidate with input scores \schoolgradesValue is defined as the inner product $(\schoolgradesWeight\cdot\schoolgradesValue)\propto\|\schoolgradesValue\|\cos(\schoolgradesWeight, \schoolgradesValue)$.
Each circle in the plot corresponds to the input scores of a candidate. Two candidates A and B have input scores with same $l_2$-norm. A has higher selection score than B under score weight vector $w_1$. However, under score weight vector $w_2$, obtained by a rotation of angle \direction from $w_1$, B has higher selection score.
}
\label{fig:angles}
\end{figure}

To identify the optimal \merit policy we could use a sample of data $\{\schoolgradesValue\}$ and follow an approach similar to~\citet{asudeh2017designing}.
For a given set of applicants with input scores $\{\schoolgradesValue\}$, it is possible to partition the vector space of \schoolgradesWeight into regions such that any vector of score weights \schoolgradesWeight within the same region leads to the same ordering of applicants, where the ordering is performed by $\{\schoolgradesWeight\cdot\schoolgradesValue\}$.
However, we would have to perform an exhaustive search of such regions, which is inefficient for high-dimensional data.
%
\iffalse
\newline
\noindent\makebox[\linewidth]{\rule{\columnwidth}{0.1pt}}
\spara{Exhaustive space partitioning~\cite{asudeh2017designing}}
\begin{enumerate}
	\item partition the vector space of \schoolgradesWeight into regions as per~\cite{asudeh2017designing};
	\item for each region:
	\begin{itemize}
		\item choose one \schoolgradesWeight in that region;
	 	\item set the threshold \scoreThreshold accordingly to make the policy calibrated;
	 	\item evaluate the value of the objective function;
	 \end{itemize}
	 \item choose the \schoolgradesWeight with the optimal objective value.
\end{enumerate}
\noindent\makebox[\linewidth]{\rule{\columnwidth}{0.1pt}}
\fi
%
Instead, we experiment with a simple grid-like search strategy: starting with vector \schoolgradesWeight = \bestWeights, we rotate the score weight vector \schoolgradesWeight by small increments towards different directions \direction, to obtain new weight vectors \schoolgradesWeight' = \rotate{\schoolgradesWeight; \direction} and evaluate the objective function in every step.
The approach is described in Algorithm~\ref{algo:searchingformerit}.

\begin{algorithm}[b]
\caption{Search for \merit policy \newline
{\bf Parameters:} \{\direction\}: set of rotation directions, $k$: number of rotations in each direction}
\label{algo:searchingformerit}
/* initialize weights */ \\
\schoolgradesWeight := \bestWeights, \optWeights := \bestWeights \\
\For{each direction \direction}
{
	\For{$i = 1..k$}{
		/* rotate the vector of score weights */ \\
		\schoolgradesWeight := \rotate{\schoolgradesWeight; \direction} \\
		/* normalize the vector */ \\
		$\schoolgradesWeight := \schoolgradesWeight / |\schoolgradesWeight|_1$ \\
		/* set the threshold \scoreThreshold to make the policy calibrated */ \\
		$\scoreThreshold = \invcumul{1 - \admissionFraction}$ \\
		/* update best policy */ \\
		\If{\optFunction{\schoolgradesWeight} > \optFunction{\optWeights}}{
			\optWeights := \schoolgradesWeight
		}
	}
}
{\bf return} \optWeights
\end{algorithm}

\subsection{\bonus policy}
\label{sec:bonusunderbasic}

Let \meritPolicy be the optimal \merit policy, having parameters \bestWeights and \scoreThreshold as defined in Equations~\ref{eq:goodmerit}.
%
%We remind that \bestWeights are the linear coefficients that predict best the performance \excellence of candidates.
%
%Moreover, \meritPolicy is calibrated to admit a fraction \admissionFraction of applicants, according to the calibration condition.
%In what follows, we consider the parameters of \meritPolicy fixed.
%
We now consider a \bonus policy \bonusPolicy that favors with $\bonusValue$ points the group \sensitive = \sensitiveGroup that is {\it disadvantaged} under \meritPolicy, assuming without loss of generality that
\[
\prob{\admission = 1 | {\sensitive = \sensitiveGroup};\meritPolicy} \leq
\prob{\admission = 1 | {\sensitive \not= \sensitiveGroup};\meritPolicy}.
\]
For the moment, let us restrict \bonusPolicy to use the same score weights \bestWeights as \meritPolicy, but score threshold $\scoreThreshold_{_\bonusValue}$ that is potentially different than the threshold \scoreThreshold used by \meritPolicy.
We lift the restriction on score weights at the end of Section~\ref{sec:bonusunderbasic}.

\mpara{Calibrating $\scoreThreshold_{_\bonusValue}$.}
All policies are required to be calibrated, i.e., adhere to condition~\ref{eq:calibration}.
Towards this end, let us consider a fixed bonus $\bonusValue\geq 0$.
The lemma below guarantees that there is an algorithm to choose its threshold $\scoreThreshold_{_\bonusValue}$ so that \bonusPolicy becomes calibrated.
\begin{lemma}
The score threshold $\scoreThreshold_{_\bonusValue}$ that makes \bonusPolicy calibrated can be found via binary search over the interval $[\scoreThreshold, \scoreThreshold + \bonusValue]$.
% \note[Giorgio]{should it be -}. -- no :) MM
\label{lemma:binarysearch}
\end{lemma}

\mpara{Optimizing \bonusValue.}
Given $\scoreThreshold_{_\bonusValue}$ obtained with the procedure described above, and \schoolgradesWeight = \bestWeights,
how do we determine \bonusValue = \optBonusValue that optimizes the objective of Problem~\ref{problem:utilitydisparity}?
%
%We remind that \bonusPolicy is meant to favor the disadvantaged group \sensitive = \sensitiveGroup with $\bonusValue \geq 0$ points.
%
Our main observation is that we can pre-determine a narrow range of values within which the optimal $\optBonusValue$ must fall.
To be precise, let $\cumul{\score = \scoreValue|\sensitive = \sensitiveGroup; \meritPolicy}$ and $\cumul{\score = \scoreValue|\sensitive \not= \sensitiveGroup; \meritPolicy}$ be the cumulative probability functions for the score \score of the two groups under \merit policy \meritPolicy.
Let $g(\cumulValue) = \invcumul{\cumulValue|\sensitive = \sensitiveGroup; \meritPolicy}$ and $h(\cumulValue) = \invcumul{\cumulValue|\sensitive \not= \sensitiveGroup; \meritPolicy}$ be the corresponding inverse functions: given a cumulative probability value \cumulValue, they return the score \score of the \cumulValue-th applicant of the respective group.
Then the optimal $\optBonusValue$ lies in the interval given below.
\begin{lemma}
$\;\;\;\;\;\;
\optBonusValue \in  \interval = [0, {h(1 - \admissionFraction) - g(1 - \admissionFraction)}]
$
\label{lemma:rangeofbonuses}
\end{lemma}
\begin{proof}
{\it First}, notice that the following condition holds for \meritPolicy.
\begin{align}
\prob{\admission = 1 | \sensitive = \sensitiveGroup} \leq \admissionFraction \leq \prob{\admission = 1 | \sensitive \not= \sensitiveGroup}
\label{cond:inequalityofproportions}
\end{align}
This follows from the assumption that \sensitive = \sensitiveGroup is the disadvantaged group and the fact that the overall proportion \admissionFraction of accepted applicants must be between the proportions of accepted applicants of the two groups $\sensitive = \sensitiveGroup$ and $\sensitive \not= \sensitiveGroup$.
It is easy to see that the proportion $\prob{\admission = 1 | {\sensitive = \sensitiveGroup}}$ of admitted candidates of the disadvantaged group is non-decreasing for increasing $\bonusValue\geq 0$, while the corresponding proportion $\prob{\admission = 1 | {\sensitive \not= \sensitiveGroup}}$ for the advantaged group is non-increasing for increasing $\bonusValue\geq 0$.
This implies that, for increasing $\bonusValue\geq 0$, there is a bonus value $\bonusValue = \bonusValue_{_\dmd}$ for which we have equality between the quantities of Condition~\ref{cond:inequalityofproportions} and $\dmd = 0$.
This is achieved for a bonus value that makes the (1 - \admissionFraction)-th applicant of group $\sensitive=\sensitiveGroup$ have a score equal to the score of the (1 - \admissionFraction)-th applicant of group $\sensitive\not=\sensitiveGroup$.
Formally, we have
$
\bonusValue_{_\dmd} = h(1 - \admissionFraction) - g(1 - \admissionFraction)
$
and the \bonus policy with $\bonusValue=\bonusValue_{_\dmd}$ is calibrated for score threshold
$
\scoreThreshold_{\bonusValue_{_\dmd}} = \invcumul{1 - \admissionFraction|\sensitive \not= \sensitiveGroup; \meritPolicy}.
$
The above implies that the absolute value $|\dmd|$ of disparity decreases over the interval $\bonusValue\in[0, \bonusValue_{_\dmd}]$ and increases over $\bonusValue\in[\bonusValue_{_\dmd}, \infty]$.
\
{\it Second}, it is easy to see that \uoa is decreasing for increasing \bonusValue.
The two observations imply that the objective function (Problem~\ref{problem:utilitydisparity}) is optimized in  $[0, \bonusValue_{_\dmd}]$. %The lemma follows.
\end{proof}

One simple algorithm to search for \optBonusValue is grid-based search, in which we evaluate the performance of \bonus(\bonusValue) policies for uniformly distributed \bonusValue in the range $\interval$.
The calibration for each policy explored in this grid-based search can be done somewhat more efficiently based on the following generalization of Lemma~\ref{lemma:binarysearch}, which allows us to narrow the range of binary search for the threshold \scoreThreshold that makes a policy calibrated. The search algorithm for \optBonusValue is shown in Algorithm~\ref{algo:searchforbonus} (proof omitted).
\begin{lemma}
For $i\in\{1,2\}$, let $\scoreThreshold_i$ be the threshold that makes the \bonus policy with $\bonusValue_i$ bonus points calibrated, and let \scoreThreshold be the threshold for \meritPolicy.
If $\bonusValue_1 \leq \bonusValue_2$, then
$
\scoreThreshold_1 \leq \scoreThreshold_2 \leq \scoreThreshold + \bonusValue_2.
$
%\label{observe:narrowerbinarysearch}
\end{lemma}

\begin{algorithm}[t]
\caption{Grid-based search for \optBonusValue\newline
{\bf Parameters:} $k$: granularity}
\label{algo:searchforbonus}
$\epsilon = \bonusValue_{_\dmd} / k$ \\
\optBonusValue := 0 \\
% $v_{_{OPT}} := \optFunction{\optBonusValue}$ \\
\For{i = 1..k}
{
	$\bonusValue = i\cdot\epsilon$ \\
	$\scoreThreshold_{i} := \binarysearch{\scoreThreshold_{i-1}, \scoreThreshold + \bonusValue}$ \\
	$v$ := $\optFunction{\bonusValue, \scoreThreshold_{_\bonusValue}}$ \\
	\If{$v > v_{_{OPT}}$}{
		\optBonusValue := \bonusValue;\ \  $v_{_{OPT}} := v$
	}
}
\bf{return} \optBonusValue
\end{algorithm}

\mpara{Solving Problem~\ref{problem:utilitydisparity}}.
So far in Section~\ref{sec:bonusunderbasic}, we have constrained ourselves to \bonus policies that adopt fixed weights \schoolgradesWeight = \bestWeights.
This raises the question: how do we choose \schoolgradesWeight jointly with other parameters to find the optimal bonus policy for Problem~\ref{problem:utilitydisparity}?
The answer is that we can simply use \schoolgradesWeight = \bestWeights, the optimal score weights for a \merit policy with $\lambda = 0$, as stated in following lemma.
\begin{lemma}
There is an optimal bonus policy with \schoolgradesWeight = \bestWeights.
\end{lemma}
\begin{proof}
{\it First}, consider two \bonus policies $\bonusPolicy_{1}$ and $\bonusPolicy_{2}$ with the following properties: (i) $\bonusPolicy_{1}$ and $\bonusPolicy_{2}$ are associated with the same value of $|\dmd|$; (ii) $\bonusPolicy_{1}$ uses score weights \schoolgradesWeight = \bestWeights, while $\bonusPolicy_{2}$ does not.
Then, $\bonusPolicy_{1}$ is associated with higher \uoa and therefore higher objective value than $\bonusPolicy_{2}$.
{\it Second}, notice that for well-behaved score distributions (specifically, for cumulative score distributions $\cumul{\score = \scoreValue|\sensitive = \sensitiveGroup; \meritPolicy}$ and $\cumul{\score = \scoreValue|\sensitive \not= \sensitiveGroup; \meritPolicy}$ that are continuous), it is possible to choose an appropriate bonus value $\bonusValue\in(-\infty, +\infty)$ to achieve {\it any} value of $\dmd\in(-1, 1)$, irrespective of score weights \schoolgradesWeight.
Combined, the above two arguments imply that for any \bonus policy $\bonusPolicy_{2}$ with $\schoolgradesWeight\not=\bestWeights$, we can build a policy with same \dmd and \schoolgradesWeight=\bestWeights, which is associated with higher objective value. %, proving the lemma.
\end{proof}

\subsection{\quota policy}

We find that there is a relationship of equivalence between \bonus and \quota policies, as pointed out in the following lemma.
\begin{lemma}
For a given \bonus policy, there is a \quota policy that leads to the selection of exactly the same candidates, and vice versa.
\end{lemma}
\begin{proof}
Let $\prob{\admission = 1, \sensitive = \sensitiveGroup}$ be the probability admitting applicants of group \sensitive = \sensitiveGroup under a \bonus policy with bonus value \bonusValue.
Let also $\quotaValue_\bonusValue$ be such that
\begin{align*}
\prob{\admission = 1, \sensitive = \sensitiveGroup} & = \quotaValue_\bonusValue\ \admissionFraction \\
\prob{\admission = 1, \sensitive \not= \sensitiveGroup} & = (1 - \quotaValue_\bonusValue)\ \admissionFraction.
\end{align*}
and define the probabilities
\begin{align*}
\sigma_{\sensitive = \sensitiveGroup} = \prob{\admission = 1\ |\ \sensitive = \sensitiveGroup} = \frac{\quotaValue_\bonusValue\ \admissionFraction}{\prob{\sensitive = \sensitiveGroup}} \\
\sigma_{\sensitive \not= \sensitiveGroup} = \prob{\admission = 1\ |\ \sensitive  \not= \sensitiveGroup} = \frac{(1 - \quotaValue_\bonusValue)\ \admissionFraction}{\prob{\sensitive \not= \sensitiveGroup}}.
\end{align*}
We can construct a \quota policy with a quota $\quotaValue = \quotaValue_\bonusValue$ for the disadvantaged group (\sensitive = \sensitiveGroup), by setting
\begin{align*}
\scoreThreshold_{_{(\sensitive     = \sensitiveGroup)}} = \invcumul{1 - \sigma_{\sensitive     = \sensitiveGroup} | ; \sensitive     = \sensitiveGroup} \\
\scoreThreshold_{_{(\sensitive \not= \sensitiveGroup)}} = \invcumul{1 - \sigma_{\sensitive \not= \sensitiveGroup} | ; \sensitive \not= \sensitiveGroup},
\end{align*}
as per the definition of \quota policies in Section~\ref{sec:admissionpolicies}.
By construction, we have a \quota policy that accepts exactly the same applicants as the given \bonus policy. A similar argument can be made for the opposite direction.
\end{proof}
This lemma allows us to focus our empirical study (Section~\ref{sec:dataanalysis}) mainly on \merit and \bonus policies.

\section{Policy search for multiple sensitive attributes}
\label{sec:multiplesensitiveattributes}

Our discussion so far has focused on policies that target one sensitive attribute.
In this section, we introduce the problem of policy design for multiple sensitive attributes $\{\sensitive_i\}$.
% Like before, we assume that there are two groups associated with each sensitive attribute, namely $(\sensitive_i = \sensitiveGroup_i)$ and $(\sensitive_i \not= \sensitiveGroup_i)$.

\begin{problem}
For sensitive properties $\sensitive_1, \sensitive_2, \ldots, \sensitive_m$, each partitioned into two associated groups of applicants $(\sensitive_i=\sensitiveGroup_i)$ and $(\sensitive_i\not=\sensitiveGroup_i)$,
define \scoringFunction{\sensitive_i, \schoolgrades} and \admissionFunction{\sensitive, \score} so as to maximize
\begin{equation*}
\optFunction{\scoringFunctionInstance,\admissionFunctionInstance} = \uoa(\scoringFunctionInstance,\admissionFunctionInstance)
- \sum\lambda_i\ |\dmd_{_\sensitive}(\scoringFunctionInstance,\admissionFunctionInstance)|
,\;\;\lambda_i \geq 0.
\end{equation*}
\label{problem:multiplesensitiveattributes}
\end{problem}

\mpara{Algorithms.}
For \merit policies, Algorithm~\ref{algo:searchingformerit} is directly extensible to the setting of multiple sensitive attributes.

For \bonus policies, our goal is to assign bonuses $\bonusValue_1$, $\bonusValue_2$, $\ldots$, $\bonusValue_m$ to the respective disadvantaged groups $(\sensitive_i=\sensitiveGroup_i)$.
One approach is, similarly to Algorithm~\ref{algo:searchforbonus}, to perform a grid search over possible bonus values for each sensitive attribute.
However: \begin{inparaenum}[(i)]\item Lemma~\ref{lemma:binarysearch} does not extend directly to multiple dimensions, and so does not binary search; \item grid search at fine granularity and multiple dimensions is inefficient\end{inparaenum}.
Instead, we employ an incremental algorithm that works as follows: \begin{inparaenum}[(i)]\item at every step, the algorithm maintains bonus values $\bonusValue_1$, $\bonusValue_2$, $\ldots$, $\bonusValue_m$ for each sensitive attribute, \item at every step, it greedily increments  by an increment $\bonusInc$ the bonus value $\bonusValue_i$ of the attribute $\sensitive_i$ which leads to the best objective value\end{inparaenum}.

For \quota policies, the equivalence with \bonus extends to multiple sensitive attributes (proof omitted), and so we do not explore them further here.

%!TeX root = admissions.tex

\section{Application}
\label{sec:dataanalysis}

In what follows, we present a large real-world dataset of university admissions (Section~\ref{sec:data}), build a causal model on it (Section~\ref{sec:building}), and evaluate the performance of affirmative action policies on the data (Section~\ref{sec:evaluation}).
In addition, we compare with two baseline approaches: (i) the \fair algorithm of \citet{zehlike2017fa}, and (ii) the \median algorithm of \citet{feldman2015certifying} and highlight interesting connections among them (Section~\ref{sec:evaluation}).

\subsection{Data}
\label{sec:data}

%In this section, we describe the datasets over which we carry out the empirical analysis.
%
We analyze anonymized data about university admissions from %an OECD country (details withheld for double-blind review).
Chile.
In Chile, admission to undergraduate programmes is typically based on standardized tests, administered by a public entity, on language, mathematics, natural sciences, and human sciences, plus their grades in high school. The latter are converted to the same range of values as the standardized tests using a simple formula. Together, test results and grades form the input scores.
%
% Our names <- Official names used in Chile
% "natural sciences" <- science
% "human sciences" <- history and social sciences

Universities announce their degrees (such as engineering, medicine, or law), and coefficients to convert the input scores into an admission score.
%
%Grades, mathematics, and language always have a non-zero coefficient, while the remaining tests might be present or not with a non-zero coefficient in the formula.
%
For instance, the engineering school we analyze below uses 30\% grades, 10\% language, 45\% mathematics, 15\% natural sciences, and 0\% human sciences. At the same university, the law school instead uses 40\% grades, 25\% language, 10\% mathematics, 0\% natural sciences, and 25\% human sciences.
\iffalse % We can skip this for now, as we do not use the data
%
Additionally, each degree is associated with a number of vacancies for applicants achieving the highest admission scores, and a few vacancies for additional programmes benefiting underrepresented groups or exceptional athletes.
\fi
After learning their test results, candidates rank degrees by decreasing order of preference. Then, applicants for each degree are ranked by admission score and the top ones are admitted.
% Once a student is admitted at one of his/her preferred degrees, the remaining preferences s/he has are discarded.

We have access to two datasets from this admission process.
The first corresponds to \emph{national-level data} about students who took the standardized tests; the second corresponds to \emph{university-level data} of students admitted to one of the most sought-after engineering degrees in one of the largest universities (95\% of the admitted applicants enroll for this degree).
The national-level data is provided by state authorities in anonymized form under a research agreement\footnote{
	Department of Measurement, Evaluation, and Registry for Education (DEMRE).
	Data requests through~\url{https://psu.demre.cl/portales/portal-bases-datos}.
}.
This dataset includes the gender, high school type (public or private), income decile, grades in high school, and results in standardized tests of about $3,500$ applicants who gave the test in 2017 and included the engineering degree of our university-level data among their preferences.
The university-level data is similarly provided by university authorities in anonymized form specifically for research on discrimination and bias.
For the university-level data, we have access to the same variables, except for income decile, for about $3,000$ admitted applicants from 2010 to 2014.
Most importantly, the university-level data contains the grades obtained by admitted students during their first year at the university.
We note that the period of the two datasets do not overlap, but due to standardization of grades and test scores, the distributions of input scores do not differ significantly across years.
%To preserve the anonymity of the students, no attempt of joining the national-level and university-level data was made.

\begin{figure}
\includegraphics[width=0.85\columnwidth]{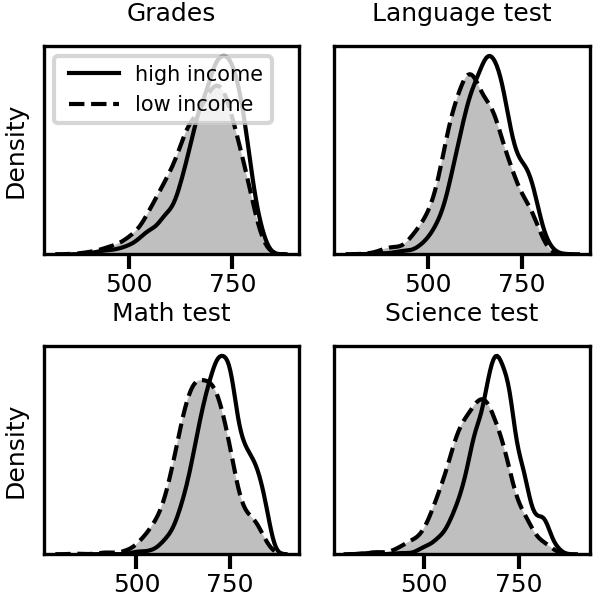}
\caption{Differences in input scores by income.}
\label{fig:incomegrades}
\end{figure}

% \includegraphics[width=1.85\columnwidth]{img/school_type_vs_grades.png}
% \caption{Differences in input scores by school type.}
% \label{fig:schooltypegrades}

\begin{figure}
\includegraphics[width=0.85\columnwidth]{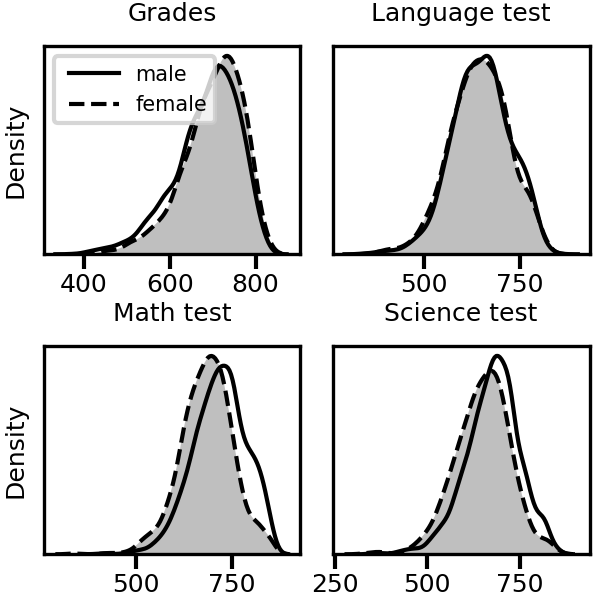}
\caption{Differences in input scores by gender.}
\label{fig:gendergrades}

\end{figure}

\subsection{Building a causal model}
\label{sec:building}

The causal model (see Figure~\ref{fig:basicdiagram}) is built from the two datasets we have at our disposal.
In the causal model, there are two conditional relationships that need to be set to implement an admission policy, namely $(\sensitive, \schoolgrades)\rightarrow\score$ and $(\score, \sensitive)\rightarrow\admission$, and two conditional relationships that need to be learnt from data, namely,
\begin{align}
	\sensitive \rightarrow \schoolgrades & : \prob{\schoolgrades\ |\ \sensitive} \label{eq:gradesconditional}\\
	\schoolgrades, \admission = 1 \rightarrow \excellence & : \prob{\excellence\ |\ \admission = 1, \schoolgrades}. \label{eq:excellenceconditional}
\end{align}
Both probabilistic quantities above can be extracted directly from the data (e.g., by fitting a statistical model on the involved attributes).
Here is a subtle point that needs to be addressed for \prob{\excellence\ |\ \admission = 1, \schoolgrades} (Expr.~\ref{eq:excellenceconditional}): if we learn a model for it from the data of students who were admitted (\admission = 1) according to a specific admission policy, can we use it to make predictions for students who {\it would be admitted} under a different policy -- or are the observed data somehow {\it biased} by the decisions of the admission policy that generated the data?
The answer is provided by causality theory~\cite{bookofwhy}, from which we deduce that we can, indeed, use the model we learn from the data to make predictions under different admission policies.
The reason is that attribute \schoolgrades is an appropriate `control' variable for Expr.~\ref{eq:excellenceconditional} (more formally, \schoolgrades blocks all back-door paths from \admission to \excellence; for an intuitive explanation see~\cite{bookofwhy}, Chapter 7).

\iffalse
\spara{Bias in observations.}
%
In general, we should not expect to predict accurately \excellence for applicants with \schoolgrades values that are far from the observed ones.
%
This is potentially a problem in our case, as the students whose performance we observe are those admitted in this highly-selective school, who tend to have high input scores \schoolgrades that are higher than those of the average candidate.
%
Nevertheless, as we will see later in this section, the alternative policies that we explore in the rest of the section do not suffer from this problem, as they do not lead to a dramatically different cohort of accepted applicants compared to the current policy for which the data were observed.
%
This is because most of the observed differences across policies involve applicants that are close to the decision margin.
\fi

In what follows, we use the attributes of gender (male or female), high school type (public or private), and income (above or below the median) as sensitive attributes \sensitive, with the population split in the two groups in parentheses; the values of grades, language test, math test, and natural sciences test as input scores \schoolgrades; and the sum of grades during first year at university as \excellence (weighted by the credits of each subject; all first-year students take the same subjects).

For Expr.~\ref{eq:gradesconditional}, we use directly the distributions extracted from the data, as shown in Figures~\ref{fig:incomegrades}-\ref{fig:gendergrades} (the plots for school type are omitted here and whenever they are similar to those for income).
%
%The density values shown in the plots were produced by the {\it seaborn}\footnote{\url{https://seaborn.pydata.org/}} visualization library.
%
As we see, the two groups for income (and similarly for school type) exhibit visibly different distributions (Fig.~\ref{fig:incomegrades}), while the distributions differ less for gender (Fig.~\ref{fig:gendergrades}).
This suggests that, for the admission system that we study here, there is a stronger case for affirmative action policies that target financial attributes than gender, and the enactment of such policies has larger potential for impact.

\begin{figure}
\begin{center}
	\includegraphics[width=.8\columnwidth]{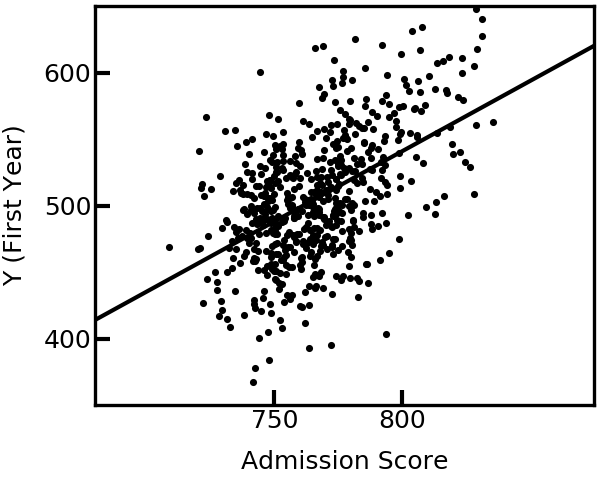}
\end{center}
\caption{First year performance vs admission score of the optimal \merit policy with $\lambda = 0$. For clarity, the plot shows only a 20\% sample of the data points, stratified by \excellence.}
\label{fig:linearfyt}
\end{figure}

For Expr.~\ref{eq:excellenceconditional}, we learn the best linear weights \bestWeights for a \merit policy that targets only \uoa (i.e., considers $\lambda = 0$) -- they are $47.6\%$ grades, $1.69\%$ language, $23.8\%$ mathematics, and $26.9\%$ natural sciences.
As discussed in Section~\ref{sec:algorithmsformerit}, these weights also correspond to the linear function that expresses the expected performance \excellence of a selected applicant with given input scores.
Figure~\ref{fig:linearfyt} shows the actually observed values of admission scores \score and performance \excellence of admitted students during their first year of studies.
There is a correlation between the two, which supports our choice of a linear function to model the expectation
$m(\schoolgradesValue) = \expect{\excellence | \schoolgrades = \schoolgradesValue, \admission = 1}$ (see also beginning of Section~\ref{sec:algorithmsformerit}).

Note also that there is variance in the observed values of \excellence for a given score value.
Such variance does not affect the task of policy design in the setting we described in Section~\ref{sec:setting}, as utility \uoa is defined in terms of expected performance \excellence.
Generally speaking, however, large variance in \excellence could trigger political arguments in favor or against affirmative action policies.
On one hand, large variance might indicate that, instead of affirmative action, policy design should focus on the design of tests, the values \schoolgrades of which ought to be better predictors of future performance. % (e.g., perhaps the test for mathematics should be re-designed to better distinguish different levels in the analytical aptitude of students).
On the other hand, variance that is too large compared to the expected performance of applicants might also support arguments in favor of affirmative action policies: if we cannot predict accurately who is going to do well at university, then we should not rely too much on test scores, but instead provide more opportunities to applicants from disadvantaged backgrounds.

Going back to our data, the distributions for admission score with weights \bestWeights are shown in Figure~\ref{fig:scoredistribution}.
It becomes obvious that, for the \merit policy that targets only \uoa, high-income applicants have an advantage over low-income applicants (similarly for private-school applicants against public-school applicants), in the sense that they are admitted by larger proportions.
Male applicants have an advantage over female applicants, but the advantage is smaller than in the case of income or school type, consistently with our earlier observations.

\begin{figure}
\begin{center}

\includegraphics[width=0.95\columnwidth]{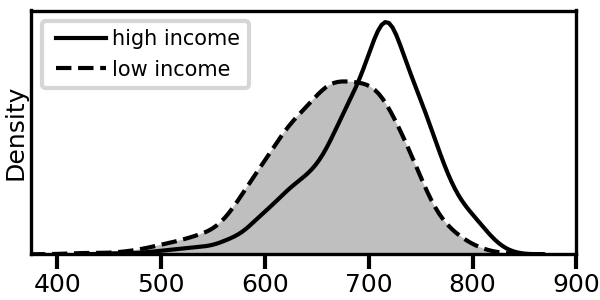}

\includegraphics[width=0.95\columnwidth]{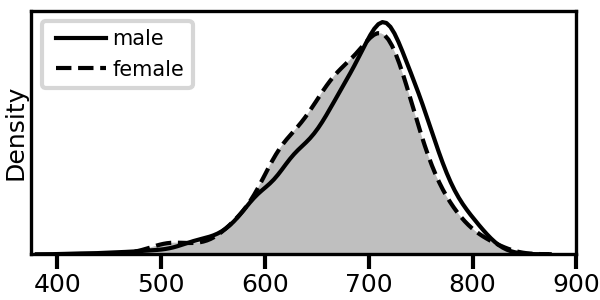}
\end{center}
\caption{Distribution of admission scores for the optimal \merit policy with $\lambda = 0$.}
\label{fig:scoredistribution}
\end{figure}

\iffalse

\begin{table}
\caption{\merit}
\label{table:weights}
\begin{tabular}{c|c|c|c|c}
\hline
$\lambda$ & Grades & Language & Math & Science \\
\hline
$0$ & 47.6\% & 1.69\% & 23.8\% & 26.9\% \\
\hline
\end{tabular}
\end{table}

\begin{table}
\caption{\merit}
\label{table:weights}
\begin{tabular}{c|c|c|c|c|c}
\hline
$\lambda$ & \sensitive & Grades & Language & Math & Science \\
\hline
$0$ & (none) & 47.6\% & 1.69\% & 23.8\% & 26.9\% \\
\hline
$100$ & income	& 49\% & 0.69\% & 9.8\% & 41\% \\
$100$ & school\_type	& 68\% & 1\% & 15\% & 17\% \\
$100$ & gender	& 68\% & 1\% & 15\% & 16\% \\
\hline
\end{tabular}
\end{table}
\fi

\begin{figure}
\begin{center}
\includegraphics[width=.95\columnwidth]{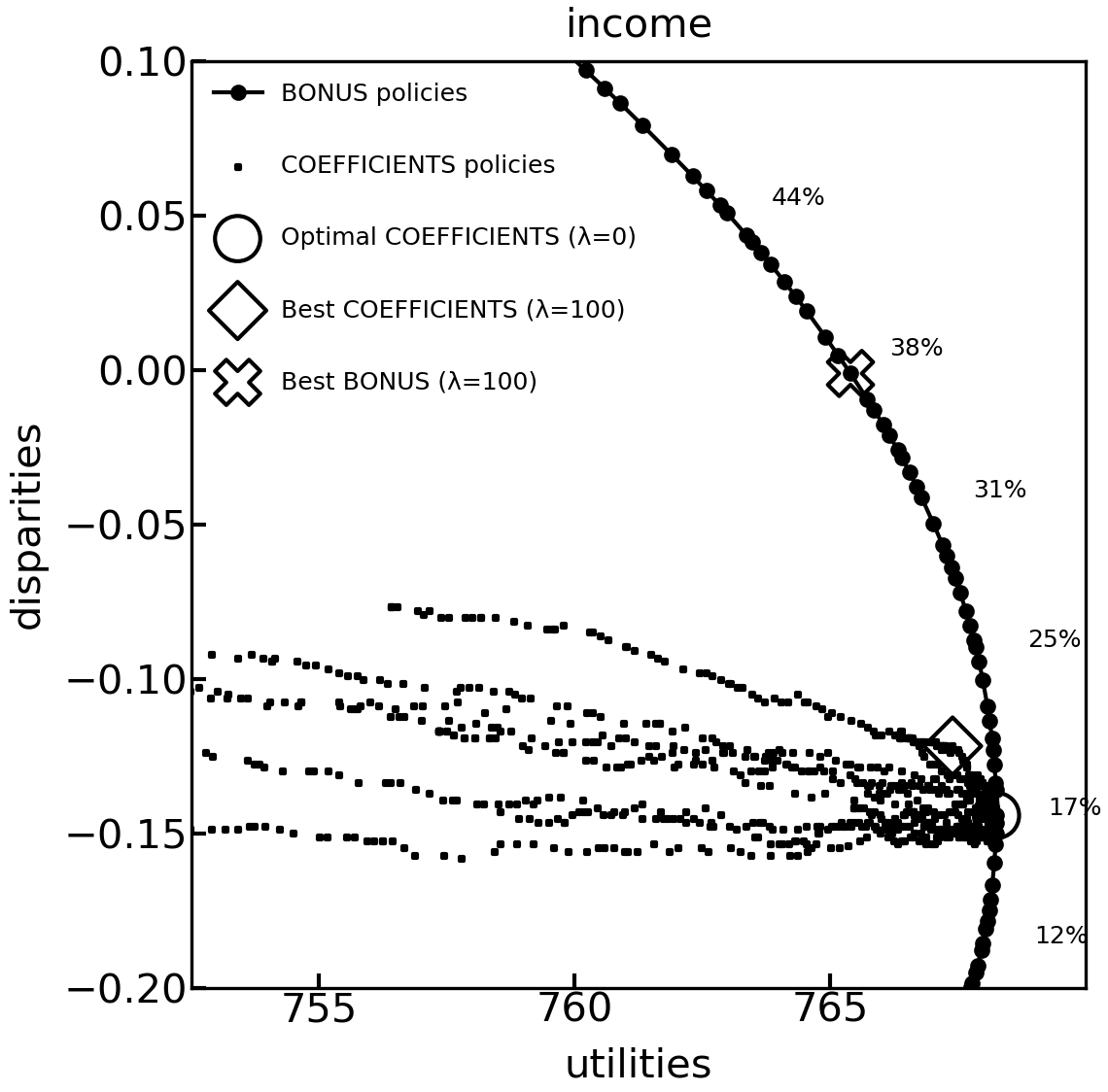}
\includegraphics[width=.95\columnwidth]{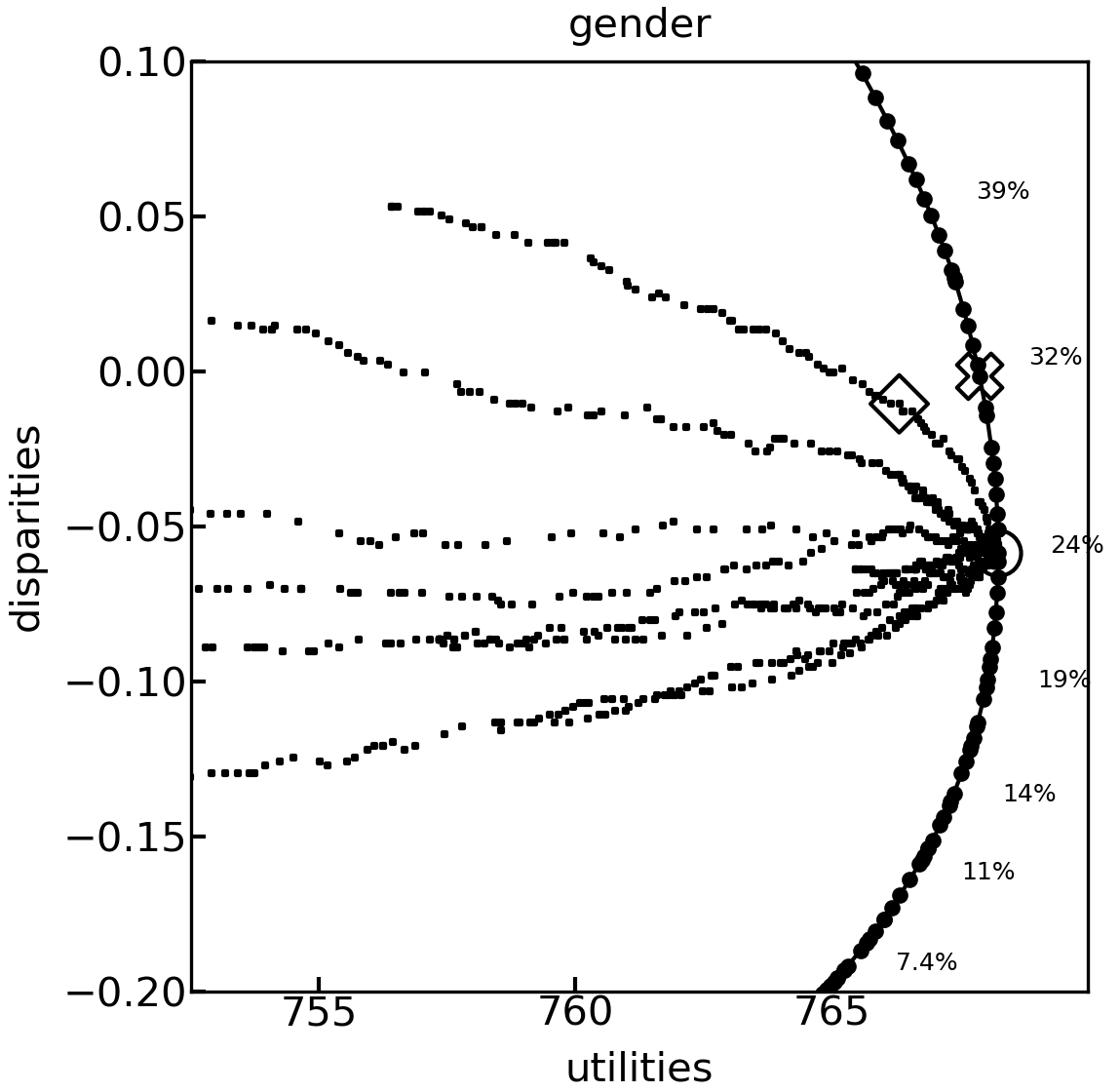}
\caption{Comparison of \merit and \bonus policies. Quotas for the disadvantaged groups under \quota policies corresponding to \bonus policies are shown on the right.}
\label{fig:singleattribute}
\end{center}
\end{figure}

\subsection{Affirmative action policies}
\label{sec:evaluation}

Figure~\ref{fig:singleattribute} shows the empirical performance of different affirmative action policies, each policy concerning a single sensitive attribute.
To produce the plot, we employed Algorithm~\ref{algo:searchingformerit} to search for alternative \merit policies and Algorithm~\ref{algo:searchforbonus} to search for different \bonus policies.
The plots also show a few indicative quota percentages for the corresponding \quota policy of the same attribute.
The plots also depict the optimal \merit policy with $\lambda = 0$, the best \merit and \bonus (and equivalent \quota) policy with $\lambda = 100$ discovered by the policy search algorithms~\ref{algo:searchingformerit}-\ref{algo:searchforbonus}, where $\lambda$ is the parameter that dictates the trade-off between \uoa and \dmd in the optimization function (see Problem~\ref{problem:utilitydisparity}).

We make a few observations.
\begin{inparaenum}
\item Based on its performance for income and school type, the \bonus and \quota policies provide better exploration of the trade-off between \uoa and \dmd: they lead to a better value of the optimization function (not shown in the plot) and policy instances with practically zero disparity \dmd, with little loss in \uoa.
\item The difference is less pronounced when gender is used as the sensitive attribute.
This is explained by the fact that student grades show smaller discrepancy across genders than across income and school types (see Figure~\ref{fig:incomegrades}-\ref{fig:gendergrades})
\item Consistently with the previous observation, the optimal \bonus policy uses larger bonus value for income and school type (\bonusValue = $30$ and $35$, respectively) than the one for gender (\bonusValue = $9.4$).
To put these numbers in context, consider that the admission scores of the top applicant is about $833$ points and of the last admitted applicant is in the range $722$ - $733$, depending on the attribute \sensitive. Hence, the gender bonus is negligible.
\item As expected, \bonus policies that lead to higher disparity (i.e., more favorable for the disadvantaged group) correspond to higher quota values for the corresponding \quota policy.
\end{inparaenum}

% \note[ChaTo]{Perhaps some commentary could be added of the sort ``The optimal bonus policy requires X points; consider that the admission scores of the first and last student admitted under the baseline are Y and Z''}

\spara{Baseline approaches}
We compare the policies discussed in this paper with policies based on alternative rankings produced by algorithms \median~\cite{feldman2015certifying} and \fair~\cite{zehlike2017fa}.
% \squishlist{

	% \item
	\median~\cite{feldman2015certifying} is applied on the scores \schoolgrades of each applicant and produces ``repaired'' feature values, such that it is impossible to distinguish from them whether an applicant belongs to group (\sensitive=\sensitiveGroup) or ($\sensitive\not=\sensitiveGroup$).
	Specifically, for each input score \schoolgrades, \median considers its distribution within each group, and produces repaired values such that an applicant of group (\sensitive=\sensitiveGroup) at the $q$-th quantile of the in-group distribution obtains the same repaired score as the corresponding applicant of group ($\sensitive\not=\sensitiveGroup$).
	This is done so that in-group order of applicants is preserved.
	For our purposes, we employ \median to obtain repaired scores \schoolgrades, use them to compute the admission score with weights \bestWeights for all applicants, rank them and admit the top.
	In all cases, the result we obtain empirically is identical to the result of \bonus for which \dmd = 0.
	This is explained theoretically: both policies admit top applicants from within each group -- and by construction, the policy based on \median admits the same proportion of applicants from each group, leading to zero disparity.

	% \item
	\fair~\cite{zehlike2017fa} takes as input the admission scores of applicants and produces a ranking that preserves the in-group rank of applicants, 
	% (i.e., better applicants of a given group are placed higher), 
	but at the same time satisfies ``ranked group fairness'' constraints.
	The latter ensures that within {\it any prefix} of the ranking both groups are sufficiently represented -- i.e., that the positions occupied by the disadvantaged group do not deviate significantly below a desired number.
	Deviation is defined in terms of a statistical test using two parameters $\alpha$ and $\rho$.
	For our purposes, we employed \fair on the admission scores obtained with weights \bestWeights, and admit the top applicants.
	We observed empirically that both \fair and \bonus led to admissions that span the same utility-vs-disparity curve (see Figure~\ref{fig:singleattribute}).
	This is explained theoretically: both policies admit top applicants from within each group and, for appropriately chosen parameters, they can lead to the same desired number of admitted applicants from each group.
	However, we note that \bonus policy does not lead necessarily to the same {\it ranking} as \fair, even when it leads to the same {\it set} of $k$ admitted applicants.
	This is because \fair enforces representation of the disadvantaged group for {\it every} prefix in the ranking, while \bonus ranks by selection score.
% }\squishend

\spara{Multiple Sensitive Attributes}
Figure~\ref{fig:multipleattributes} shows the results when all three attributes (income, school type and gender) are considered.
Again, we mark the optimal \merit policy for $\{\lambda_i = 0\}$, as well as the optimal \merit and \bonus (and equivalent \quota) policy with $\{\lambda_i = 100\}$, where $\{\lambda_i\}$ are the parameters that dictate the trade-off between \uoa and \dmd in the optimization function (see Problem~\ref{problem:multiplesensitiveattributes}).
Consistently with the case of single sensitive attributes, the optimal bonus values assigned to income and school type (\bonusValue = $16$ and $26$, respectively) are larger than the one for gender (\bonusValue = $7$). We note that the baseline  \fair cannot handle multiple sensitive attributes.

\begin{figure}
\centering\includegraphics[width=.92\columnwidth]{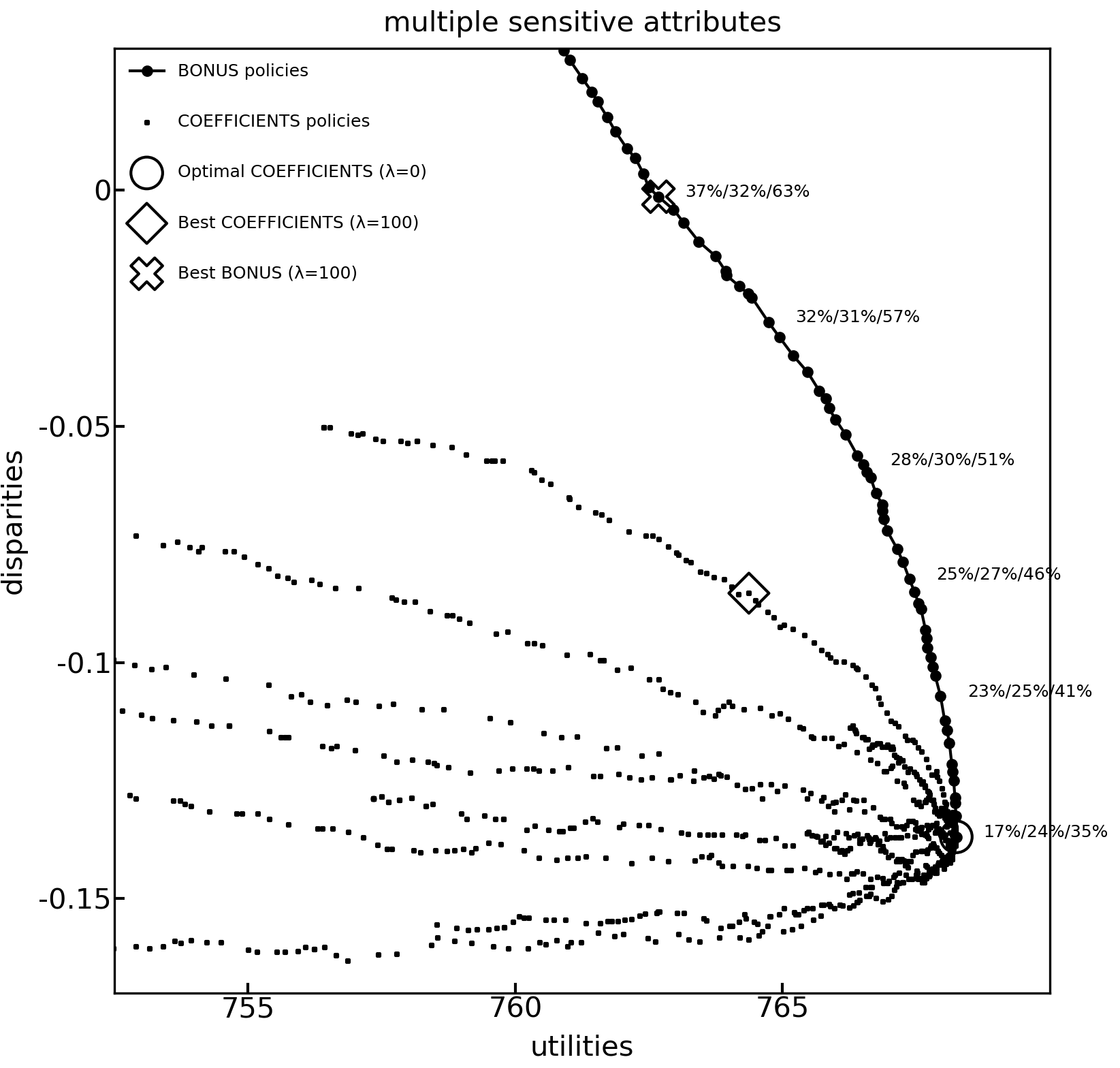} \caption{Performance of admission policies for multiple sensitive attributes. Percentages on the right are corresponding quotas for the disadvantaged groups (by income, gender, and school type) under \quota policies matching sample \bonus policies.}
\label{fig:multipleattributes}
\end{figure}

\iffalse

	\begin{figure}
	\begin{center}
	\includegraphics[width=\columnwidth]{img/explore_coeffs_gender.png}
	\includegraphics[width=\columnwidth]{img/explore_coeffs_income.png}
	\includegraphics[width=\columnwidth]{img/explore_coeffs_school_type.png}
	\caption{\merit policies}
	\end{center}
	\end{figure}

	\begin{figure}
	\begin{center}
	\includegraphics[width=\columnwidth]{img/explore_bonus_gender.png}
	\includegraphics[width=\columnwidth]{img/explore_bonus_income.png}
	\includegraphics[width=\columnwidth]{img/explore_bonus_school_type.png}
	\caption{\bonus policies}
	\end{center}
	\end{figure}

	\begin{figure}
	\begin{center}
	\includegraphics[width=\columnwidth]{img/joint_bonus.png}
	\caption{Joint \bonus policies}
	\end{center}
	\end{figure}

\fi

% \todo[MM]{
% 	Distribute \bonusValue points over multiple sensitive attributes and optimize generalized objective.
% }
% \todo[MM]{
% 	Compare current policy vs best alternatives.
% }

% \input{extensions}
%!TeX root = admissions.tex

\section{Conclusions}
\label{sec:conclusions}

In this paper, we demonstrated how the problem of designing affirmative action policies can be addressed in an algorithmic framework.
We explored three types of policies, i.e., \merit, \bonus, and \quota.
We showed the equivalence of \bonus and \quota policies, and provided efficient search algorithms for the optimal parameters of \merit and \bonus policies.
We analyzed real data for university admissions and obtained instances of admission policies that achieved good utility (performance of admitted applicants) while virtually eliminating disparity (i.e., the difference of the proportions of admitted applicants from advantaged and disadvantaged groups).
% Future steps include the application of the proposed algorithms on a wider range of real settings and the treatment of other measures of fairness.

\noindent
\section*{Acknowledgments} 
This work was supported by the HUMAINT project of the European Commission's Joint Research Centre for Advanced Studies in Seville. 
C. Castillo was partially funded by La Caixa project LCF/PR/PR\-16/11110009.
G. Barnabo was partially supported by ERC Advanced Grant  788893 AMDROMA "Algorithmic and Mechanism Design Research in Online Markets".
S. Celis was partially funded by Complex Engineering Systems Institute (CONICYT-PIA-FB0816).

%\clearpage
% \balance
\bibliographystyle{ACM-Reference-Format}
\bibliography{biblio}
%\balancecolumns % GM June 2007

\end{document}